\documentclass[11pt,a4paper]{article}

\usepackage{subfigure}
\usepackage{epic,eepic}
\usepackage{graphicx,color}

\usepackage{float}
\usepackage{ifthen}
\usepackage{xspace}

\usepackage{amssymb}
\usepackage{amsmath}
\usepackage{amsfonts} 
\usepackage{stmaryrd} 
\usepackage{latexsym} 

\usepackage{a4wide}

\newtheorem{notation}{Notation}

\newtheorem{definition}{Definition}
\newtheorem{theorem}{Theorem}
\newenvironment{remark}{\noindent\textit{Remark:}}{}

\newcommand\set[1]{\ensuremath{\{ #1 \} }}
\renewcommand\int[1]{\ensuremath{\llbracket #1 \rrbracket}}
\newcommand\cur[1]{\ensuremath{{\cal{#1}}}}

\newcommand\commentaire[1]{}

\usepackage{url}

\begin{document}

\begin{center}

{\LARGE \bf
Termination of Multipartite Graph Series\\\smallskip Arising from Complex Network Modelling}

\bigskip

Matthieu Latapy
\ \ \ \ \
Thi Ha Duong Phan
\ \ \ \ \
Christophe Crespelle\,\footnote{Corresponding author: christophe.crespelle@inria.fr}
\ \ \ \ \
Thanh Qui Nguyen

\bigskip
\bigskip

\begin{minipage}{.9\textwidth}
\centerline{\bf Abstract.}
\medskip
An intense activity is nowadays devoted to the definition of models capturing the properties of complex networks. Among the most promising approaches, it has been proposed to model these graphs via their clique incidence bipartite graphs. However, this approach has, until now, severe limitations resulting from its incapacity to reproduce a key property of this object: the overlapping nature of cliques in complex networks. In order to get rid of these limitations we propose to encode the structure of clique overlaps in a network thanks to a process consisting in iteratively \emph{factorising} the maximal bicliques between the upper level and the other levels of a multipartite graph. We show that the most natural definition of this factorising process leads to infinite series for some instances. Our main result is to design a restriction of this process that terminates for any arbitrary graph. Moreover, we show that the resulting multipartite graph has remarkable combinatorial properties and is closely related to another fundamental combinatorial object. Finally, we show that, in practice, this multipartite graph is computationally tractable and has a size that makes it suitable for complex network modelling.
\end{minipage}

\end{center}

\section{Introduction}\label{sec-intro}


It appeared recently \cite{watts98collective,albert02statistical,dorogovtsev02evolution}
that most real-world complex networks (like the internet topology, data exchanges, web graphs, social networks, or biological networks) have some non-trivial properties in common. In particular, they have a very low density, low average distance and diameter, an heterogeneous degree distribution, and a high local density (usually captured by the clustering coefficient \cite{watts98collective}).
Models of complex networks aim at reproducing these properties.
Random\,\footnote{In all the paper, {\em random} means uniformly chosen in a given class.} graphs with given numbers of vertices and edges \cite{erdos59random}
fit the density and distance properties, but they have homogeneous degree distributions and low local density. Random graphs with prescribed distributions \cite{molloy95critical}
and the preferential attachment model \cite{barabasi99emergence} fit the same requirement, with the degree distribution in addition, but they still have a low local density. As these models are very simple, formally and computationnaly tractable, and rather intuitive, there is nowadays a wide consensus on using them.

However, when one wants to capture the high local density in addition to previous properties, there is no clear solution. In particular, we are unable to construct a random graph with prescribed degree distribution and local density. As a consequence, many proposals have been made, e.g. \cite{watts98collective,dorogovtsev02evolution,physicaa06guillaume,ipl04guillaume}
, each with its own advantages and drawbacks.
Among the most promising approaches, \cite{ipl04guillaume,physicaa06guillaume} propose to model complex networks based on the properties of their clique incidence bipartite graph (see definition below).
They show that generating bipartite graphs with prescribed degree distributions for bottom and top vertices and interpreting them as clique incidence graphs results in graphs fitting all the complex network properties listed above, including heterogeneous degree distribution and high local density.

However, the bipartite model suffers from severe limitations. In particular, it does not capture overlap between cliques, which is prevalent in practice.
Indeed, as evidenced in \cite{ipl04guillaume,socialnetworks07latapy}, the neighbourhoods of vertices in the clique incidence bipartite graph of a real-world complex network generally have significant intersections: cliques strongly overlap and vertices belong to many cliques in common.
On the opposite, when one generates a random bipartite graph with prescribed degree distributions, the obtained bipartite graph have much smaller neighbourhood intersections, almost always limited to at most one vertex (under reasonable assumptions on the degree distributions).
Indeed, the process of generation based on the bipartite graph is equivalent to randomly choosing sets of vertices of the graph (with prescribed size distribution) that we all link together. Because of the constraints imposed on this size distribution by the low density of the graph , the probability of choosing several vertices in common between two such random sets tends to zero when the graph grows.
As a consequence, the bipartite model fails in capturing the overlapping nature of cliques in complex networks. This leads in particular to graphs which have many more edges than the original ones (two cliques of size $d$ lead to $d.(d-1)$ edges in the model graph, while the overlap between cliques make this number much smaller in the original graph).

\subsubsection*{Our contribution}

Since the random generation process of the bipartite graph is not able to generate non-trivial neighbourhood intersections (that is having cardinality at least two), a natural direction to try to solve this problem consists in using a structure explicitly encoding these intersections. This can be done using a tripartite graph instead of a bipartite one: one may encode any bipartite graph $B = (\bot,\top,E)$ into a tripartite one $T = (\bot, \top, C, E')$ where $C$ is the set of non-trivial maximal bicliques (complete bipartite graphs having at least two bottom vertices and two top vertices) of $B$ and $E'$ is obtained from $E$ by adding the edges between any biclique $c$ in $C$ and all the vertices of $B$ which belong to $c$ and removing the edges between vertices of $C$. This process, which we call \emph{factorisation}, can be iterated to encode any graph in a multipartite one where there are hopefully no non-trivial neighbourhood intersections.

In this paper, we show that this iterated factorising process do not end for some graphs. We then introduce variations of this base process and study them with regard to termination issue. Our main result is the design of such a process, which we call \emph{clean factorisation}, that terminates on any arbitrary graph. In addition, we show that the multipartite graph on which terminates this process has remarkable combinatorial properties and is strongly related to a fundamental combinatorial object. Namely, its vertices are in bijection with the chains of the inf-semilattice of intersections of maximal cliques of the graph.
Finally, we give an upper bound on the size and computation time of the graph on which terminates the iterated clean factorising process of $G$, under reasonable hypothesis on the degree distributions of the clique incidence bipartite graph of $G$; therefore showing that this multipartite graph can be used in practice for complex network modelling.

\subsubsection*{Outline of the paper}

We first give a few notations and basic definitions useful in the whole paper.
We then consider the most immediate generalisation of the bipartite decomposition (Section~\ref{sec-wfs})
 and show that it leads to infinite decompositions in some cases. We propose a more restricted version in Section~\ref{sec-fs}, which seems to converge but for which the question remains open. Finally, we propose another restricted version in Section~\ref{sec-cfs} for which we prove that the decomposition scheme always terminates.

\subsubsection*{Notations and preliminary definitions}

All graphs considered here are finite, undirected and simple (no loops and no multiple edges). A graph $G$ having vertex set $V$ and edge set $E$ will be denoted by $G=(V,E)$. We also denote by $V(G)$ the vertex set of $G$. The edge between vertices $x$ and $y$ will be indifferently denoted by $xy$ or $yx$.

A $k$-partite graph $G$ is a graph whose vertex set is partitioned into $k$ parts, with edges between vertices of different parts only (a bipartite graph is a $2$-partite graph, a tripartite graph a $3$-partite graph, etc):
$G=(V_0,\ldots,V_{k-1},E)$ with $E \subseteq \set{uv\ |\ u \in V_i, v \in V_j, i\not=j}$.
The vertices of $V_i$, for any $i$, are called the {\em $i$-th level} of $G$, and the vertices of $V_{k-1}$ are called its \emph{upper vertices}.

$\cur{K}(G)$ denotes the set of maximal cliques of a graph $G$, and $N^G(x)$ the neighbourhood of a vertex $x$ in $G$. When $G=(V_0,\ldots,V_{k-1},E)$ is $k$-partite, we denote by $N_i^G(x)$, where $0\leq i\leq k-1$, the set of neighbours of $x$ at level $i$: $N_i^G(x)=N^G(x)\cap V_i$. When the graph referred to is clear from the context, we omit it in the exponent.
A \emph{biclique} of a graph is a set of vertices of the graph inducing a complete bipartite graph.
We denote $B(G)$ the clique incidence graph of $G=(V,E)$, {\em i.e.} its bipartite decomposition: $B(G)=(V,\cur{K}(G),E')$ where $E' = \{vc \ |\ c \in \cur{K}(G), \ v \in c\}$.

In all the paper, an operation will play a key role, we name it \emph{factorisation} and define it generically as follows.

\begin{definition}[factorisation]
Given a $k$-partite graph $G=(V_0,\ldots,V_{k-1},E)$ with $k\geq 2$ and a set $V'_k$ of subsets of $V(G)$, we define the factorisation of $G$ with respect to $V'_k$ as the $(k+1)$-partite graph $G'=(V_0,\ldots,V_k,(E\setminus E_-)\cup E_+)$ where:
\begin{itemize}
\item $V_k$ is the set of maximal (with respect to inclusion) elements of $V'_{k}$,
\item $E_-=\set{yz\ |\ \exists X\in V_{k}, y\in X\cap V_{k-1} \text{ and } z\in X\setminus V_{k-1}}$, and
\item $E_+=\set{Xy\ |\ X\in V_{k} \text{ and } y\in X}$.
\end{itemize}
When $V_k\not=\emptyset$, the factorisation is said to be \emph{effective}.
\end{definition}

In the rest of the paper, we will refine the notion of factorisation by using different sets $V'_k$ on which is based the factorisation operation, and we will study termination of the graph series resulting from each of these refinements.

The converse operation of the factorisation operation is called \emph{projection}.

\begin{definition}[projection]
Given a $k$-partite graph $G=(V_0,\ldots,V_{k-1},E)$ with $k\geq 3$, we define the projection of $G$ as the $(k-1)$-partite graph $G'=(V_0,\ldots,V_{k-2},(E\cap (\bigcup_{1\leq i\leq k-2}V_i)^2)\cup A_+)$ where $A_+=\set{yz\ |\ \exists i,j \in \int{1,k-2}, i\neq j \text{ and } y\in V_i \text{ and } z\in V_j \text{ and } \exists t\in V_{k-1}, yt,zt \in E}$ is the set of edges between any pair of vertices of $\bigcup_{1\leq i\leq k-2}V_i$ having a common neighbour in $V_{k-1}$.
\end{definition}

It is worth to note that the projection is the converse of the factorisation operation independently from the set $V'_k$ used in the definition of the factorisation.

\section{Weak factor series}\label{sec-wfs}

As explained before, our goal is to improve the bipartite model of~\cite{ipl04guillaume,physicaa06guillaume} in order to be able to encode non-trivial clique overlaps, that is overlaps whose cardinality is at least two. Since these overlaps in the graph result from the neighbourhood overlaps of the upper vertices, the purpose of the new model we propose is to encode the graph into a multipartite one by recursively eliminating all non-trivial neighbourhood overlaps of the upper vertices. We first describe this process informally, then give its formal definition and exhibit an example for which it does not terminate.

Neighbourhood overlaps of the upper vertices in a bipartite graph $B=(V_0,V_1,E)$ may be encoded as follows. For any maximal\footnote{The reason why one would take the maximal bicliques is simply to try to encode all neighbourhood overlaps using a reduced number of new vertices. Notice that there are other ways to reduce even more the number of new vertices created, for example by taking a biclique cover of the edge set of $B$. This is however out of the scope of this paper.} biclique $C$ that involves at least two upper vertices and two other vertices, we introduce a new vertex $x$ in a new level $V_2$, add all edges between $x$ and the elements of $C$, and delete all the edges of $C$, as depicted on Figure~\ref{fig-multipartite}. We obtain this way a tripartite graph $T = (V_0,V_1,V_2,E')$ which encodes $B$ (one may obtain $B$ from $T$ by the projection operation) and which has no non-trivial neighbourhood overlaps in its first two level ($V_0$ and $V_1$).

\begin{figure}[!h]
\centering
\includegraphics[scale=0.33]{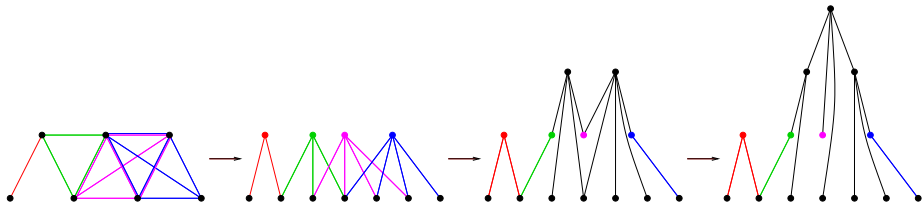}
\caption{Example of multipartite decomposition of a graph. From left to right: the original graph; its bipartite decomposition; its tripartite decomposition; and its quadripartite decomposition, in which there is no non-trivial neighbourhood overlap anymore. In this case, the decomposition process terminates.}
\label{fig-multipartite}
\end{figure}

This process, which we call a \emph{factorising step}, may be repeated on the tripartite graph $T$ obtained (as well as on any multipartite graph) by considering the bipartite graph between the upper vertices and the other vertices of the tripartite (or multipartite) graph, see Figure~\ref{fig-multipartite}.
All $k$-partite graphs obtained along this iterative factorising process have no non-trivial neighbourhood overlap between the vertices of their $k-1$ first levels. Then, the key question is to know whether the process terminates or not.



\smallskip

We will now formally define the factorising process and show that it may result in an infinite sequence of graphs. In the following sections, we will restrict the definition of the factorising step in order to always obtain a finite representation of the graph.

\begin{definition}[$V^\bullet_{k}$ and weak factor graph]\label{weakfactor}
Given a $k$-partite graph $G=(V_0,\ldots,V_{k-1},E)$ with $k\geq 2$, we define the set $V^\bullet_{k}$ as:
$$V^\bullet_{k}=\set{\set{x_1,\ldots,x_l}\cup\bigcap_{1\leq i\leq l} N(x_i)\ |\ l \ge 2,\ \forall i\in\int{1,l}, x_i\in V_{k-1} \text{ and } |\bigcap_{1\leq i\leq l} N(x_i)|\geq 2}.$$
The \emph{weak factor graph} $G^\bullet$ of $G$ is the factorisation of $G$ with respect to $V^\bullet_{k}$.
\end{definition}

The weak factorisation admits a converse operation, called projection, which is defined in Section~\ref{sec-intro}.
It implies that the factor graph of $G$, as well as its iterated factorisations, is an encoding of $G$.

The \emph{weak factor series} defined below is the series of graphs produced by recursively repeating the weak factorising step.

\begin{definition}[weak factor series $\cur{WFS}(G)$]
The weak factor series of a graph $G$ is the series of graphs $\cur{WFS}(G)=(G_i)_{i\geq 1}$ in which $G_1=B(G)$ is the clique incidence graph of $G$ and, for all $i\geq 1$, $G_{i+1}$ is the weak factor graph of $G_i$: $G_{i+1} = G^\bullet_{i}$. If for some $i\geq 1$ the weak factor operation is not effective then we say that the series is {\em finite}.
\end{definition}

\begin{figure}[!h]
\centering
\includegraphics[scale=0.5]{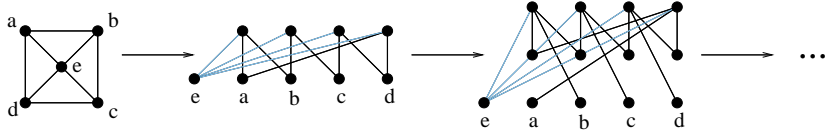}
\caption{An example graph for which the weak factorising process is infinite. From left to right: the original graph $G$, its bipartite decomposition $B(G)$, and its tripartite decomposition $B(G)^\bullet$. The shaded edges are the ones involving vertex $e$, which play a special role: all the vertices of the upper level of the decompositions are linked to $e$. The structure of the tripartite decomposition is very similar to the one of the bipartite decomposition, revealing that the process will not terminate.}
\label{fig-infinite}
\end{figure}

Figure~\ref{fig-multipartite} gives an illustration for this definition. In this case, the weak factor series is finite. However, this is not true in general; see Figure~\ref{fig-infinite}. Intuitively, this is due to the fact that a vertex may be the base for an infinite number of factorising steps (like vertex $e$ in the example of Figure~\ref{fig-infinite}). The aim of the next sections is to avoid this case by giving more restrictive definitions.

\section{Factor series}\label{sec-fs}

In the previous section, we have introduced weak factor series which appear to be the most immediate extension of bipartite decompositions of graphs. We showed that, unfortunately, weak factor series are not necessarily finite.
In this section, we introduce a slightly more restricted definition that forbids the repeated use of a same vertex to produce infinitely many factorisations (as observed on the example of Figure~\ref{fig-infinite}). However, we have no proof that it necessarily gives finite series, which remains an open question.

\begin{definition}[$V^\circ_{k}$ and factor graph]\label{factor}
Given a $k$-partite graph $G=(V_0,\ldots,V_{k-1},E)$ with $k\geq 2$, we define the set $V^\circ_{k}$ as:
$$V^\circ_{k}\ =\ \set{X\in V^\bullet_k \text{ such that } |\bigcap_{y\in X\cap V_{k-1}} N_{k-2}(y)| \geq 2}.$$
The \emph{factor graph} $G^\circ$ of $G$ is the factorisation of $G$ with respect to $V^\circ_{k}$.
\end{definition}

This new definition results from the restriction of the weak factor definition by considering only sets $X\in V^\bullet_k$ such that the vertices of $X\cap V_{k-1}$ have at least two common neighbours at level $k-2$.
In this way, the creation of new vertices depends only on the edges between levels $k-1$ and $k-2$ (even though some other edges may be involved in the factorisation operation). Thus, a vertex will not be responsible for infinitely many creations of new vertices.
This restriction also plays a key role in the convergence proof of the \emph{clean factor series}, defined in next section. That is why we think it may be possible that it is sufficient to guarantee the convergence of the factor series, but we could not prove it with this sole hypothesis.



\section{Clean factor series}\label{sec-cfs}

In the two previous sections, we studied two multipartite decompositions of graphs. The first one is very natural but it does not lead to finite objects. The second one remains very general but we were unable to prove that it leads to finite object. As a first step towards this goal, we introduce here a more restricted definition for which we prove that the decomposition is finite. This new combinatorial object has many interesting features, and we consider it worth of study in itself. In particular, we prove that it is a decomposition of a well-known combinatorial object: the inf-semi-lattice of the intersections of maximal cliques of $G$. This correspondence allows to calculate quantities of graph $G$ from elements of $M$. One of such results is an explicit formula (not presented here) giving the number of triangles in $G$, which is a very important parameter of complex networks.

The \emph{clean factor graph} (defined below) is a proper restriction of the factor graph in which the vertices at level $k-1$ used to create a new vertex at level $k$ are required to have exactly the same neighbourhoods at all levels strictly below level $k-2$, except at level $1$. Intuitively, this requirement implies that the new factorisations push further the previous ones and are not simply a rewriting at a higher level of a factorisation previously done. The particular role of level $1$ will allow us to differentiate vertices of the multipartite graph by assigning them sets of nodes at level $0$. Let us now formally define the clean factor graph and its corresponding series.

\begin{definition}[$V^*_{k}$ and clean factor graph]\label{CleanFactor}
Given a $k$-partite graph $G=(V_0,\ldots,V_{k-1},E)$ with $k\geq 4$, we define the set $V^*_{k}$ as:
$$V^*_{k}\ =\ \set{X\in V^\circ_k\ |\ \forall x,y\in X\cap V_{k-1},\forall p\in\set{0}\cup\int{2,k-3}, N_p(x)=N_p(y)}.$$
The \emph{clean factor graph} $G^*$ of $G$ is the factorisation of $G$ with respect to $V^*_{k}$.
\end{definition}

\begin{definition}[clean factor series $\cur{CFS}(G)$]
The clean factor series of a graph $G$ is the series of graphs $\cur{CFS}(G)=(G_i)_{i\geq 1}$ in which $G_1=B(G)$ is the clique incidence graph of $G$, $G_2=G^\circ_1$, $G_3=G^\circ_2$ and, for all $i\geq 3$, $G_{i+1}$ is the clean factor graph of $G_{i}$: $G_{i+1} = G^*_{i}$. If for some $i$ the clean factor operation is not effective then we say that the series is {\em finite}.
\end{definition}

The rest of this section is devoted to proving the following theorem.

\begin{theorem}\label{CFSstop}
For any graph $G$, the clean factor series $(G_i)_{i\geq 1}$ is finite.
\end{theorem}

%

\begin{notation}\label{ThickClean}
Let $(G_i)_{i\geq 1}$ be the clean factor series of $G$. For any $i\geq 1$, any $x\in V_i$ and any $j<i$, we denote by $V_j(x)$ the set $N^{G_i}_j(x)$ and by $V(x)$ the set $\bigcup_{0\leq j< i} V_j(x)$.
\end{notation}

\begin{remark}
In the rest of the paper, when referring to Definition~\ref{CleanFactor}, it is worth keeping in mind that for $x\in V_{k-1}$ and $p\leq k-3$, the sets $N_p(x)$ and $N_p(y)$ used in the definition are precisely the sets $V_p(x)$ and $V_p(y)$.
\end{remark}

\begin{definition}\label{DefOK}
We denote by $\cur{O}'$ the set $\set{O\subseteq V(G)\ |\ \exists k\geq 2, \exists C_1,\ldots ,C_k\in \cur{K}(G), (\forall j,l\in\int{1,k}, j\neq l \Rightarrow C_j\neq C_l) \text{ and } O=\bigcap_{1\leq i\leq k} C_i}$; and by $\cur{O}$ the set $\set{O\in\cur{O}'\ |\ |O|\geq 2}$.
For any $O\in\cur{O}'$, we denote by $K(O)$ the set $\set{C\in \cur{K}(G)\ |\ O\subseteq C}$.
We also denote by $\cur{C}$ the set $\set{Y\subseteq \cur{K}(G)\ |\ \exists O\in\cur{O}', Y=K(O)}$.
\end{definition}

It is clear from the definition that $\cur{O}'$ is closed under intersection, this is also the case for $\cur{C}$.

\commentaire{
\begin{proposition}\label{PteOK}
For any $A,B\in V(G)$, $K(A)\cap K(B)=K(A\cup B)$.
Conversely, if $A_1,\ldots,A_n,C\in\cur{O}'$, with $n\geq 2$, and $\bigcap_{1\leq i\leq n} K(A_i)=K(C)$ then $\bigcup_{1\leq i\leq n} A_i\subseteq C$.
\end{proposition}
\begin{proof}
Let $A,B\in V(G)$.
The cliques in $K(B)\cap K(C)$ are exactly the cliques that contain both $A$ and $B$, {\em i.e.} the cliques that contain $A\cup B$. Therefore $K(A)\cap K(B)=K(A\cup B)$.

Let $A_1,\ldots,A_n,C\in\cur{O}'$, with $n\geq 2$, such that $\bigcap_{1\leq i\leq n} K(A_i)=K(C)$. From what precedes, $\bigcap_{1\leq i\leq n} K(A_i)=K(\bigcup_{1\leq i\leq n} A_i)$. Consequently, we have $K(\bigcup_{1\leq i\leq n} A_i)=K(C)$. Then, $\bigcup_{1\leq i\leq n} A_i\subseteq \bigcap_{y\in K(\bigcup_{1\leq i\leq n} A_i)} V_0(y)=\bigcap_{y\in K(C)} V_0(y)$. On the other hand, since $C\in\cur{O}'$, $C=\bigcap_{y\in K(C)} V_0(y)$. It follows that $\bigcup_{1\leq i\leq n} A_i\subseteq C$.
\end{proof}
} 


In all the $G_i$'s of the clean factor series, vertices at level $0$ correspond to vertices of $G$, vertices at level $1$ correspond to the maximal cliques of $G$, that is for any $y\in V_1, V_0(y)\in\cur{K}(G)$. That is the reason why in the following we do not distinguish between the elements of $\cur{K}(G)$ and those of $V_1$.
We will show that the vertices of $V_2$ correspond to the elements of $\cur{O}$.
Indeed, $x\mapsto V_0(x)$ is a bijection from $V_2$ to $\cur{O}$.
First, for any $x\in V_2$, by definition, $|V_0(x)|\geq 2$, then $V_0(x)=\bigcap_{y\in V_1(x)} V_0(y)$ belongs to $\cur{O}$.
Let $O\in\cur{O}$. Let us show that $X=K(O)\cup\bigcap_{y\in K(O)} V_0(y)$ is a maximal element of $V^\circ_2$. First note that $X\cap V_0=O$ and then $|X\cap V_0|\geq 2$. Now, if you augment $X$ with an element of $y\in V_1\setminus K(O)$, since $y\not\in K(O)$, $X\cap V_0$ will decrease. Thus $X$ is maximal and there is a corresponding $x\in V_2$ such that $V_0(x)=O$.
Furthermore, it is straightforward to see that the maximality of $V(x)$ implies that $V_1(x)=K(O)$. Which proves the uniqueness of the $x\in V_2$ such that $V_0(x)=O$.

\begin{definition}\label{DefCharSeq}
Let $G$ be a graph and let $(G_i)_{i\geq 1}$ be its clean factor series.
The \emph{characterising sequence} $S(x)=(O_1(x),\ldots ,O_{k-1}(x)$ of a vertex $x\in V_k$, with $k\geq 2$, is defined by:
\begin{itemize}
\item $O_1(x)=V_0(x)$
\item $\forall j\in\int{2,k-1}, O_j(x)$ is the unique element\footnote{By convention, $O_j(x)=V(G)$ when $\bigcap_{y\in V_j(x)} V_1(y)=\varnothing$.} of $\cur{O'}$ such that $K(O_j(x))=\bigcap_{y\in V_j(x)} V_1(y)$.
\end{itemize}
\end{definition}

Note that $O_j$ is properly defined. Indeed, since $\cur{C}$ is closed under intersection, a simple recursion would show that for all $i\geq 3$ and for all $y\in V_i $, $V_1(y)=\bigcap_{z\in V_{i-1}} V_1(z)\in \cur{C}$.

\smallskip

Theorem~\ref{ThCharSeq} is our main combinatorial tool for proving the finiteness of the clean factor series (Theorem~\ref{CFSstop}). Its proof is rather intricate, but it gives much more information than the finiteness of the series. By associating a sequence of sets to each vertex in levels greater than $V_2$ in the multipartite graph, we show that each such vertex corresponds to a chain of the inf-semi-lattice of the intersections of maximal cliques of $G$. The correspondence thereby highlighted between this very natural structure and the multipartite factorisation scheme we introduced is non-trivial and of great combinatorial interest.

\begin{theorem}
\label{ThCharSeq}
Let $G$ be a graph and $(G_i)_{i\geq 2}$ its clean factor series. We then have the following properties:
\begin{enumerate}
\item\label{strict} $\forall k\geq 2$, $\forall x\in V_k$, $O_1(x)\subsetneq \ldots \subsetneq O_{k-1}(x)$ and if $k=3$, $O_2(x)\in \cur{O}$ and if $k\geq 4$, $(O_2(x), \ldots , O_{k-2}(x))\in \cur{O}^{k-3}$ 

\item\label{inj} $\forall k\geq 2$, $\forall x,y\in V_k$, $x\neq y \Rightarrow S(x)\neq S(y)$,
\item\label{surj} $\forall k\geq 2$, $\forall (O_1,\ldots ,O_{k-1})\in \cur{O}^{k-1}$, $O_1\subsetneq \ldots \subsetneq O_{k-1} \Rightarrow \exists x\in V_k, S(x)=(O_1,\ldots ,O_{k-1})$.
\end{enumerate}
\end{theorem}

For lack of space, we do not give the proof of Theorem~\ref{ThCharSeq}. It can be made by recursion on $k$. The key of our proof is that we could characterise, for any $k\geq 3$, the vertices at level $k-1$ involved in the creation of a new vertex $x$ at level $k$ : roughly, they are those vertices $y$ such that there exist $O_1,\ldots ,O_{k-3},O_m,O_M \in \cur{O}$ and $S_y=(O_1,\ldots ,O_{k-3},O_k-2(y))$ is such that $O_m\subseteq O_{k-2}(y)\subseteq O_M$. Then, the characterising sequence of the created vertex $x$ is $S(x)=(O_1,\ldots ,O_{k-3},O_m,O_M)$.  
Please refer to the webpages of the authors for a complete version of the paper including proof of Theorem~\ref{ThCharSeq}.



Theorem~\ref{CFSstop} is a corollary of Theorem~\ref{ThCharSeq}. Indeed, Theorem~\ref{ThCharSeq} states that the characterising sequence $(O_1(x),\ldots ,O_{k-1}(x))$ of any node $x$ at level $k$ is such that $O_1(x)\subsetneq \ldots \subsetneq O_{k-1}(x)$. The strict inclusions imply that the length of the characterising  sequence, which is equal to $k-1$, cannot exceed the height $h$ of the inclusion order of elements of $\cur{O}$. Since $h\leq n-1$, necessarily $V_{n+1}$ is empty. It follows that the clean factor series is finite and stops at rank at most $n$.

\paragraph{Size of the multipartite model}

The size of the multipartite graph $M$ obtained at termination of the clean factor series can be exponential in theory, as the number of maximal cliques itself may be exponential. But in practice, its size is quite reasonable and it can be computed efficiently.
Theorem~\ref{th:size} below shows that under reasonable hypotheses, the size of $M$ only linearly depends on the number of vertices of $G$, with a multiplicative constant reflecting the complexity of imbrication of maximal cliques.

\begin{theorem}\label{th:size}
If every vertex of $G$ is involved in at most $k$ maximal cliques and if every maximal clique of $G$ contains at most $c$ vertices, then $|V(M)|\leq 4\times min(k\,2^c\,c!\, , 2^k\,k!)\times n$.
\end{theorem}

This upper bound can be obtained by bounding the number of sequences $O_1,\ldots ,O_i$ in two different ways: either by consedering sequences ending with a fixed set $O_i=A$, which are obtained by starting from set $A$ and removing vertices one by one; or by considering sequences starting with a fixed set $O_1=B$, which are obtained by starting from a maximal clique containing $B$ and intersecting it by one more maximal clique containing $B$ at each step.


In practice, parameters $k$ and $c$ are quite small, as they are often constrained by the context itself independently from the size of the graph. Then, the size of $M$ is small. An important consequence is that, using algorithms enumerating the cliques or bi-cliques of a graph (see~\cite{Gely20091447} for a recent survey), $M$ can be computed efficiently, that is in low polynomial time, since the number of maximal cliques is small.



\section{Perspectives}\label{sec-conclu}



Many questions arise from our work. The first one is to find minimal restrictions of the factorising process that guarantee termination. On the other hand, for processes that do not always terminate, one may determine on which classes of graphs those processes terminate. Another question of interest is the termination speed, as well as the size of the obtained encoding: proving upper bounds with softer hypothesis would be desirable.


Finally, the use of multipartite decompositions as models of complex networks, in the spirit of the bipartite decomposition, asks for several questions. In this context, the key issue is to generate a random multipartite graph while preserving the properties of the original graph. To do so, one has to express the properties to preserve as functions of basic multipartite properties (like degrees, for instance) and to generate random multipartite graphs with these properties. This is a promising direction for complex network modelling, but much remains to be done.

\medskip
\noindent
{\bf Acknowledgements.}
We warmly thank
Jean-Loup Guillaume,
Stefanie Kosuch
and
Cl\'emence Magnien
for helpful discussions.


\small{
\bibliographystyle{plain}
\bibliography{xbib,perso}
}
\end{document}